\documentclass[copyright]{eptcs}
\providecommand{\event}{Infinity 2013} 
\usepackage{breakurl}             

\usepackage{amsmath}
\usepackage{amssymb}
\usepackage{amsthm}
\usepackage{nicefrac}
\usepackage{graphicx}
\usepackage{tikz}
\usepackage[all]{xy}
\usepackage{hyperref}
\usepackage{subfig}

\newtheorem{definition}{Definition}
\newtheorem{remark}{Remark}
\newtheorem{lemma}{Lemma}

\newtheorem{theorem}{Theorem}
\newtheorem{corollary}{Corollary}

\newboolean{useComments}
\setboolean{useComments}{false}
\ifthenelse{\boolean{useComments}}{%
        \newcommand{\todo}[1]{\textcolor{olive}{ TODO: {#1}}}
	\newcommand{\ms}[1]{{\color{red}\texttt{ M S: #1 :M S }}}
	\newcommand{\js}[1]{{\color{orange}\texttt{ J S: #1 :J S }}}
}{%
	\newcommand{\todo}[1]{}
	\newcommand{\ms}[1]{}
	\newcommand{\js}[1]{}
}

\graphicspath{{.},{./figures/}}

\def\De{\Delta}
\def \disjunctunion {\stackrel{.}{\cup}}
\def \aPA {(S,Act,T,s_0)}
\def \anotherPA {(S',Act,T',s'_0)}
\def \stateSet {S}
\def \stateSetDisjunct {S^\circ}
\def \startState {s_0}
\def \actionSet {Act}
\def \someequiv {\equiv_{iso}}
\def \transitionRelation {T}
\def \transitionRelationB {T'}

\newcommand{\Disc}[1] {Dist(#1)}
\def \N{\mathbb{N}}
\def \R{\mathbb{R}}
\def \extremalp {\mathbb{E}}

\def\titlerunning{Lattice structures for PA}
\def\authorrunning{J. Schuster and M. Siegle}

\title{Lattice structures for bisimilar Probabilistic Automata}
\author{Johann Schuster
\institute{University of the Federal Armed Forces Munich\\
Neubiberg, Germany}
\email{johann.schuster@unibw.de}
\and
Markus Siegle
\institute{University of the Federal Armed Forces Munich\\
Neubiberg, Germany}
\email{markus.siegle@unibw.de}
}

\begin{document}

\maketitle

{\bf Abstract}
\ms{editiert}
The paper shows that there is a deep structure on certain sets of bisimilar Probabilistic Automata (PA).
The key prerequisite for these structures is a notion of compactness of PA. 
It is shown that compact bisimilar PA form lattices. These results are then
used in order to establish normal forms (in the sense of \cite{tacas:13}) not only for finite automata, but also for infinite automata, as long as they are compact. 
\js{OK}


\js{TODO: define the lifting for equivalence relations -- so far it is defined for arbitrary relations, which is problematic.}

\section{Introduction}

Probabilistic automata (PA) \cite{segala:95,segala:95b,lynch2003compositionality}
are a powerful and popular modelling formalism, since they
allow one to reason about the behaviour of systems which feature both randomness and nondeterminism.
For probabilistic automata (PA), several notions of simulations and bisimulations have been defined in the literature \cite{segala:95,parma2007logical}.
Bisimulation relations -- in general, but also in particular for probabilistic automata -- are employed
for characterising equivalent behaviour.
Thus they may serve as the basis for checking whether two systems are equivalent in some sense.
As a straightforward consequence, bisimulation relations are
also very valuable for reducing the size of a system, by replacing it with an equivalent but
smaller one. 
Hereby the goal is to find the smallest possible bisimilar system, i.e.\ the minimal one.
We concentrate on two notions of bisimulation for PA:
strong probabilistic bisimulation and weak probabilistic bisimulation.
For automata with finite sets of states and transitions, both are known to be decidable in polynomial time \cite{segala:02,hermanns:12}.

Recently, the question of calculating
minimal canonical forms (i.e.\ normal forms) of probabilistic automata
has been tackled (again, for automata with finite sets of states and transitions) \cite{tacas:13},
where it turned out that this problem can also be solved in polynomial time.
In the present paper, we go one step further:
We show that finiteness is not required for defining minimal canonical forms.
We point out that also for automata with countably infinite (``countable'' for short) state space,
countable set of actions, and possibly uncountable set of transitions, there are notions of
normal forms.
However, we show that, in contrast to the finite case, normal forms do not always exist.
It rather turns out that an auxiliary condition, namely compactness \cite{panangaden:10}, is crucial
for the existence of normal forms.

In the context of PA, the distributions reached by probabilistic schedulers form convex sets.
The first ones to use this observation were Segala and Cattani who developed a decision algorithm
out of this fact \cite{segala:02}.
In this paper, we combine
the strongly geometric ideas of \cite{segala:02} with
the ideas of compact automata of \cite{panangaden:10}.
This enables us to extend the recent results of \cite{tacas:13} to
a class of PA with countable state space.
In this way we may use geometrical ideas to show that normal forms of PA arise naturally as some
``generating points'' -- in a strong or weak sense -- of convex sets.
However, we also show (by a counterexample) that in general one cannot 
expect normal forms for arbitrary PA with countable state space.
For normal forms wrt.~strong bisimulation, we may directly use a classical result from functional analysis, the theorem of Krein-Milman \cite{Krein1940},
that directly extends the ideas of \cite{segala:02} to the case of countable state spaces.
For weak bisimulation, we extend the results of
\cite{tacas:13} by adding some compactness assumptions.

This paper is, to the best of our knowledge, the first approach to derive lattice structures on bisimilar objects. 
Technically, we work with partially ordered sets of bisimilar PA
on which lattice structures are established, such that
the normal form corresponds to the bottom element of the lattice.
We show that there are unique bottom elements in the lattices we define.
We work on quotients of PA, but do not address the question of how to find such quotients for arbitrary automata.
For the countably infinite case, ideas from \cite{kucera:04,forejt:12} could be used to calculate
quotients, but we do not investigate this further.
However, even if it is hard to find quotients for infinite state systems, we feel that the lattice structures we establish are 
interesting just from an abstract point of view.



The paper is organised as follows: 
Sec.~\ref{sec:prelim} recalls some basic facts on preorders, lattices and probabilistic automata, compact sets and extreme points.
It also defines the notion of compact PA which is essential for our paper.
Operations on bisimilar quotients and sets of quotients
(intersection, union, rescaling)
are defined in Sec.~\ref{sec:operations}.
The core of the paper consists of the results on lattice structures given in Sec.~\ref{sec:structures}.
Some illustrating examples are provided in Sec.~\ref{sec:examples}, and Sec.~\ref{sec:conclusion} concludes the paper.

\section{Preliminaries}
\label{sec:prelim}

\begin{definition}
The disjoint union of two sets $S_1$ and $S_2$ is defined 
as $S_1 \disjunctunion S_2 := \cup_{i \in \{1,2\}} \{ (x,i) | x \in S_i \}$.
\end{definition}
The disjoint union is defined up to isomorphism,
which implies commutativity
$S_1 \disjunctunion S_2 = S_2 \disjunctunion S_1$.
There are canonical embeddings
$S_1\rightarrow S_1 \disjunctunion S_2$, $x \mapsto (x,1)$ and $S_2\rightarrow S_1 \disjunctunion S_2$, $x \mapsto (x,2)$.

\subsection{Partial orders and lattices}

\begin{definition}
A \emph{partial order} is a binary relation $\leq$ over a set $S$ which is \emph{antisymmetric}, 
\emph{transitive}, and either \emph{reflexive} or \emph{irreflexive}, i.e., for all $a$, $b$, and $c$ in $S$, we have that:
\begin{itemize}
    \item If $a \leq b$ and $b \leq a$ then $a = b$ (antisymmetry).
    \item If $a \leq b$ and $b \leq c$ then $a \leq c$ (transitivity).
    \item Either: $a \leq a$ (reflexivity) for all $a \in S$, or: $a \not \leq a$ (irreflexivity) for all $a \in S$.
\end{itemize}
A set with a partial order is called a \emph{partially ordered set} (also called a poset). 
For a poset we write
$a<b$ iff $a\leq b$ and $a\neq b$ (and similar $a>b$).
\end{definition}

\begin{definition}
Let $(S,\leq)$ be a poset, and let $a$ and $b$ be two elements in $S$. 
An element $c$ of $S$ is the \emph{meet} (or \emph{greatest lower bound} or \emph{infimum}) 
of $a$ and $b$, if the following two conditions are satisfied:
\begin{itemize}
    \item $c \leq a$ and $c\leq b$ (i.e., $c$ is a \emph{lower bound} of $a$ and $b$).
    \item For any $d \in S$, such that $d\leq a$ and $d \leq b$, we have $d \leq c$ (i.e., $c$ is greater than or equal to any other lower bound of $a$ and $b$).
\end{itemize}
An element $c$ of $S$ is the \emph{join} (or \emph{least upper bound} or \emph{supremum}) 
of $a$ and $b$, if the following two conditions are satisfied:
\begin{itemize}
    \item $a \leq c$ and $b\leq c$ (i.e., $c$ is a \emph{upper bound} of $a$ and $b$).
    \item For any $d \in S$, such that $a\leq d$ and $b \leq d$, we have $c \leq d$ (i.e., $c$ is smaller than or equal to any other upper bound of $a$ and $b$).
\end{itemize}
\end{definition}

\begin{remark}
If there is a meet (join) of $a$ and $b$, then indeed it is unique, 
since if both $c$ and $c'$ are greatest lower bounds (least upper bounds) of $a$ and $b$, 
then $c\leq c'$ and $c'\leq c$, 
whence indeed $c'=c$.
\end{remark}

\begin{definition}
A poset $(L, \leq)$ is a \emph{lattice} if it satisfies the following two axioms.
\begin{itemize}
  \item (Existence of binary meets)
     For any two elements $a$ and $b$ of $L$, the set $\{a, b\}$ has a meet: $a \land b$ 
     (also known as the greatest lower bound, or the infimum). 
  \item (Existence of binary joins)
     For any two elements $a$ and $b$ of $L$, the set $\{a, b\}$ has a join: $a \lor b$ 
     (also known as the least upper bound, or the supremum).
\end{itemize}
A lattice is called \emph{bounded}, if it has a least and a greatest element, i.e.~elements
$l,g$ such that for all $x \in L:$ $x\leq g$ and $l \leq x$. We will also write $\bot(L)=l$, $\top(L)=g$.
A lattice is called \emph{complete}, if meet and join exist for \emph{all} subsets $A\subseteq L$.
\end{definition}

We will use the following simple but elementary lemma:
\begin{lemma}[descending chain condition (DCC)]
\label{lem:equivalences}
Let $(S,\leq)$ be a poset, then the following statements are equivalent:
\begin{itemize}
  \item Every nonempty subset $A\subseteq S$ contains an element minimal in $A$
  \item $S$ contains no infinite descending chain $a_0 > a_1 > a_2 > \ldots$
\end{itemize}
\end{lemma}

\subsection{Probabilistic Automata}


First we define the notion of discrete subdistribution and related terms and notations:
\begin{definition}[(Sub-)distributions]
Let $S$ be a countable set.
A mapping $\mu: S \rightarrow [0,1]$ is called (discrete) \emph{subdistribution}, if
$\sum_{s\in S}\mu(s)\leq 1$. 
As usual we write $\mu(S')$ for $\sum_{s\in S'}\mu(s)$.
The \emph{support} of $\mu$ is defined as $Supp(\mu):=\{s\in S|\mu(s)>0\}$.
The empty subdistribution $\mu_\emptyset$ is defined by $Supp(\mu_\emptyset)=\emptyset$.
The \emph{size} of $\mu$ is defined as $|\mu|:=\mu(S)$. A subdistribution $\mu$ is
called \emph{distribution} if $|\mu|=1$. The sets $Dist(S)$ and $Subdist(S)$ denote
distributions and subdistributions defined over the set $S$.
Let $\De_s\in Dist(S)$ denote the \emph{Dirac} distribution on $s$, i.e.~$\De_{s}(s)=1$.
For two subdistributions $\mu$, $\mu'$ the sum $\mu'':=\mu\oplus \mu'$ is defined
as $\mu''(s):=\mu(s)+\mu'(s)$ (as long as $|\mu''|\leq 1$).
As long as $c\cdot |\mu|\leq 1$, we denote by $c\mu$ the subdistribution 
defined by $(c\mu)(s):=c\cdot \mu(s)$. For a subdistribution $\mu$ and a state $s\in Supp(\mu)$
we define $\mu-s$ by 
$$(\mu-s)(t)=\begin{cases}
               \mu(t) & \text{ for }t\neq s \\
               0      & \text{ for }t = s
             \end{cases}$$
\end{definition}

\begin{definition}[Lifting of relations on states to distributions]
\label{def:lifting}
Whenever there is an equivalence relation $R\subseteq S\times S$, we may lift it to $Dist(S)\times Dist(S)$ 
in the following way.
For $\mu$, $\gamma\in Dist(S)$ we write $\mu L(R) \gamma$ (or simply, by abuse of notation, $\mu R \gamma$) if and only
if for each $C\in \nicefrac{S}{R}: \mu(C)=\gamma(C)$.
\end{definition}

\begin{definition}[cf.~\cite{segala:02}]
\label{def:pa}
A probabilistic automaton (PA) $P$ is a tuple $\aPA$, where
$\stateSet$ is a countable set of \emph{states}, $\startState \in
\stateSet$ is the \emph{initial} state, $\actionSet$ is a countable set
of \emph{actions} ($\actionSet=H\disjunctunion E$, $H$ hidden actions, $E$ external actions) 
and $\transitionRelation \subseteq \stateSet \times \actionSet \times \Disc{\stateSet}$ is a \emph{transition
  relation} (can be uncountable). Whenever $(s,a,\mu)\in T$ we also write $s\stackrel{a}{\rightarrow}\mu$.
\end{definition}
In this paper we restrict ourselves to the case where $H=\{\tau\}$, i.e.~$E=Act\setminus \{\tau\}$.
Note that by the countability of $S$ it is clear that every distribution over $S$ has at most countable support.



\subsubsection{Weak transitions}

In the following we use the definitions and terminology of \cite{Segala:07}, but we leave out the definitions for labelled transition systems.
The only major difference is that we do not assume \emph{finite branching}, i.e.~for each state
$s$ the set 
$\{(a, \mu)\in Act\times Dist(S)| s\stackrel{a}{\rightarrow}\mu \}$
does not have to be finite.
Given a transition $tr = (s, a, \mu)$, we denote $s$ by $source(tr)$ and $\mu$ 
by $\mu_{tr}$.
An execution fragment of a PA $P=\aPA$ is a finite or infinite
sequence $\alpha = q_0 a_1 q_1 a_2 q_2\cdots$ of alternating states and actions, starting with a state
and, if the sequence is finite, ending in a state, where each $(q_i,a_{i+1},\mu_{i+1})\in T$
and $\mu_{i+1}(q_{i+1})>0$.
State $q_0$, the first state of $\alpha$, is denoted by $fstate(\alpha)$. If $\alpha$ is a finite
sequence, then the last state of $\alpha$ is denoted by $lstate(\alpha)$.
An \emph{execution} of $P$ is an execution fragment (of $P$) where $q_0=s_0$.
We let $frags(P)$ denote the set of execution fragments of $P$ and $frags^\ast(P)$
the set of finite execution fragments of $P$. Similarly, we let $execs(P)$
denote the set of executions of $P$ and $execs^\ast(P)$ the set of finite executions.
Execution fragment $\alpha$ is a \emph{prefix} of execution fragment $\alpha'$, denoted
$\alpha \leq \alpha'$, if sequence $\alpha$ is a prefix of sequence $\alpha'$.

The \emph{trace} of an execution fragment $\alpha$,
written $trace(\alpha)$, is the
sequence of actions obtained by restricting $\alpha$ to the set of external actions, i.e.~$Act\setminus \{\tau\}$.
For a set $E$ of executions of a PA $P$, $traces(E)$
is the set of traces of the executions in $E$. 
We say that $\beta$ is a trace of a PA $P$ if there is an execution $\alpha$ of $P$ with
$trace(\alpha)=\beta$. 
Let $traces(P)$ denote the set of traces of $P$. 

A \emph{scheduler} for a PA $P$ is a function $\sigma: frags^\ast(P)\rightarrow SubDisc(T)$
such that $tr\in supp(\sigma(\alpha))$ implies that $source(tr)=lstate(\alpha)$.
This means that the image $\sigma(\alpha)$ is a \emph{discrete} subdistribution over transitions.
The defect of the subdistribution, i.e.~$1-|\sigma(\alpha)|$ is used for stopping in the current state. In other words,
a scheduler is the entity that resolves nondeterminism in a probabilistic automaton by choosing randomly either
to stop or to perform one of the transitions that are enabled from the current state.
A scheduler $\sigma$ is said to be \emph{deterministic} if for each finite execution fragment
$\alpha$ either $\sigma(\alpha)(T)=0$ or $\sigma(\alpha)=\Delta(tr)$ (Dirac measure for $tr$) for some $tr\in T$.
A scheduler is called \emph{memoryless}, if it depends only on the last state of its argument, that is, for each pair $\alpha_1$,
$\alpha_2$ of finite execution fragments, if $lstate(\alpha_1)=lstate(\alpha_2)$, then $\sigma(\alpha_1)=\sigma(\alpha_2)$.

A scheduler $\sigma$ and a discrete initial probability measure $\mu_0\in Dist(S)$ induce a measure $\epsilon$ on the sigma-field
generated by cones of execution fragments as follows. If $\alpha$ is a finite execution fragment, then the \emph{cone}
of $\alpha$ is defined by $C_\alpha=\{\alpha' \in frags(P)| \alpha \leq \alpha' \}$.
The measure $\epsilon$ of a cone $C_\alpha$ is defined recursively:
If $\alpha=s$ for some $s\in S$ we define $\epsilon(C_\alpha)=\mu_0(s)$.
If $\alpha$ is of the
form $\alpha' a' s'$ it is defined by the equation
$$\epsilon(C_\alpha)=\epsilon(C_{\alpha'})\cdot \sum_{tr \in T(a')}\sigma(\alpha')(tr)\mu_{tr}(s'),$$
where $T(a')$ denotes the set of transitions of $T$ that are labelled by $a'$. Standard measure theoretical arguments
ensure that $\epsilon$ is well defined. We call the measure $\epsilon$ a probabilistic execution fragment of $P$, and we say
that $\epsilon$ is generated by $\sigma$ and $\mu_0$.

Consider a probabilistic execution fragment $\epsilon$ of a PA $P$, with first state $s$, i.e.~$\mu_0=\Delta(s)$, that assigns probability
1 to the set of all finite execution fragments $\alpha$ with trace $trace(\alpha) = \beta$ for some $\beta \in (Act \setminus \{\tau\})^\ast$. Let $\mu$ be the discrete
measure defined by $\mu(s')=\epsilon(\{\alpha | lstate(\alpha)=s'\})$.
Then $s\stackrel{\beta}{\Rightarrow}_C\mu$ is a \emph{weak combined transition} of $P$. We call $\epsilon$ a \emph{representation}
of $s\stackrel{\beta}{\Rightarrow}_C\mu$. If $s\stackrel{\beta}{\Rightarrow}_C\mu$ is induced by a deterministic scheduler, we also write
$s\stackrel{\beta}{\Rightarrow}\mu$.
In case $trace(\alpha)$ is empty we write $s\stackrel{\tau}{\Rightarrow}_C\mu$.

Let $\{s\stackrel{a}{\rightarrow}\mu_i\}_{i\in I}$ be a collection of transitions of a PA $P$, and let $\{c_i\}_{i\in I}$ be a collection
of probabilities such that $\sum_{i\in I}c_i=1$. Then the triple
$(s,a,\sum_{i\in I}c_i\mu_i)$ is called a \emph{(strong) combined transition} of $P$
and we write $s\stackrel{a}{\rightarrow}_C \sum_{i\in I}c_i \mu_i$.
\ms{vorigen Satz ergänzt, da strong combined notation nirgends eingefuehrt war}

\subsection{Bisimulations}


\begin{definition}[Strong probabilistic bisimulation \cite{segala:95b}]
\label{def:strong}
An equivalence relation $R$ on the set of states $S$ of a PA $P=(S, Act, T, s_0)$ is called \emph{strong probabilistic bisimulation} if and only if
$x \mathrel{R} y$ implies for all $a \in Act$: $(x\stackrel{a}{\rightarrow}\mu)$ implies $(y\stackrel{a}{\rightarrow}_C\mu')$ 
with $\mu(C)=\mu'(C)$ for all $C \in \nicefrac{S}{R}$. Two PA are called \emph{strongly bisimilar} if
their initial states are related by a strong probabilistic bisimulation relation on the direct sum of their states.
\end{definition}

\begin{definition}[Weak probabilistic bisimulation \cite{segala:95b}]
\label{def:naiveweak}
An equivalence relation $R$ on the set of states $S$ of a PA $P=(S, Act, T, s_0)$ is called \emph{weak probabilistic bisimulation} if and only if
$x \mathrel{R} y$ implies for all $a \in Act$: $(x\stackrel{a}{\rightarrow}\mu)$ implies $(y\stackrel{a}{\Rightarrow}_C\mu')$
with $\mu(C)=\mu'(C)$ for all $C \in \nicefrac{S}{R}$. Two PA are called \emph{weakly bisimilar} if
their initial states are related by a weak probabilistic bisimulation relation on the direct sum of their states.
\end{definition}

Note that (as we always use combined transitions), in the following we generally omit the word ``probabilistic'' for our bisimulation relations,
even if in the sense of \cite{segala:95b} we speak about strong probabilistic and weak probabilistic bisimulations.
In the sequel we denote, as usual, by $\sim$ a strong bisimulation relation and by $\approx$ a weak bisimulation relation.

It has been shown in \cite{segala:95b} that decision of bisimilarities is closely related to convex sets of reachable distributions,
i.e.~$S_\sim (s,a):=\nicefrac{\{\mu \in Dist(S) | s\stackrel{a}{\rightarrow}_C\mu\}}{\sim}$
and
$S_\approx (s,a):=\nicefrac{\{\mu \in Dist(S) | s\stackrel{a}{\Rightarrow}_C\mu\}}{\approx}.$
Those sets are considered modulo the bisimilarities known and splitters are constructed out of them. For details we refer to \cite{segala:95b}.
The important point is that those sets can be defined for infinite state systems as well. The only difference is that they can no longer
be regarded as subsets of $\mathbb{R}^n$ then.

It is clear by definition that the unreachable parts of an automaton do not play any role for bisimilarity, which motivates the following definition.

\begin{definition}[Reachable states]
Let $P=\aPA$ be a PA, $\stateSet' \subseteq \stateSet$ its set of reachable states, i.e.~those states that can be reached with
non-zero probability by a scheduler starting from $\startState$.
Let $\transitionRelation':=\transitionRelation|_{\stateSet'\times Act \times \Disc{\stateSet'}}$
 be the restriction of the transition relation to $\stateSet'$.
We define $r(P):=(\stateSet',Act, \transitionRelation', \startState)$ and call it the \emph{reachable fraction} of $P$.
\end{definition}



\subsection{Isomorphic \& quotient automata}

We want to be able to identify automata that are basically the same, 
only having different names of their states. The following definition formalises this.

\begin{definition}[Isomorphic automata]
Let $P=\aPA$ and $\stateSet'$ a set with $|\stateSet'|=|\stateSet|$. We call a bijective mapping between sets
$\iota: S \rightarrow S'$ a (set-)\emph{isomorphism}. 
Via this isomorphism we may push forward distributions on $S$ to distributions on $S'$
by the mapping $\mu\circ \iota^{-1}$, as the following diagram shows.
$$
\xymatrix{
  S \ar[d]^{\mu} \ar[r]^{\iota} & S' \ar@{-->}[dl]^{\mu\circ \iota^{-1}} \\
  [0,1]
}
$$
That in turn means that we may push forward transitions between states in the set $S$ to
transitions between states in the set $S'$ using our isomorphism by the following mapping
$$
\begin{array}{lccc}
\iota: & S \times Act \times \Disc{S} & \rightarrow & S' \times Act \times \Disc{S'} \\
         & (s,a,\mu) & \mapsto & (\iota(s),a,\mu\circ \iota^{-1})
\end{array}
$$
By this, we may define the automaton $\iota(P):=(\iota(S),Act,\iota(T),\iota(s_0))$
where we, as usual, denote for mappings $f:A\rightarrow B$ by $f(A):=\{f(a)| a\in A\}\subseteq B$
the \emph{image} of $A$ under $f$.
Two PA $P=\aPA$, $P'=\anotherPA$ are called \emph{isomorphic}, if there exists 
an isomorphism $\iota: S \rightarrow S'$ such that $P'=\iota(P)$. For isomorphic
automata we write $P\equiv_{iso}P'$.
\end{definition}

We would like to stress that we require the \emph{same} set of actions for two
automata to be isomorphic. The reason for this is that also for bisimulations the \emph{same} set 
of actions is considered.


\begin{lemma}
The relation $\equiv_{iso}$ is an equivalence relation.
\end{lemma}

\begin{proof}
Follows by direct verification from the properties of isomorphisms
\end{proof}

Whenever an equivalence relation $R$ is given ($R$ will be chosen to be $\sim$ or $\approx$),
it is common to look at a special automaton that is defined over equivalence classes of states. 
The following definition formalises this.

\begin{definition}[Quotient automaton]
Let $P=\aPA$ be a PA 
and $R$ an equivalence relation over 
$\stateSet$.
We write $\nicefrac{P}{R}$ to denote the 
quotient automaton of $P$ wrt.~$R$, that is
$$\nicefrac{P}{R}=(\nicefrac{\stateSet}{R},Act,\nicefrac{\transitionRelation}{R},[\startState]_R)$$
with $\nicefrac{\transitionRelation}{R}\subseteq \nicefrac{\stateSet}{R}\times Act \times \Disc{\nicefrac{\stateSet}{R}}$ such that
$([s]_R,a,\mu)\in \nicefrac{\transitionRelation}{R}$ if and only if there exists a state
$s' \in [s]_R$ such that $(s',a,\mu')\in \transitionRelation$ and $\forall [t]_R\in \nicefrac{\stateSet}{R}: \mu([t]_R)=\sum_{t'\in [t]_R}\mu'(t')$.
We call an automaton a \emph{quotient wrt.~ $R$} if it holds that
$P\equiv_{iso}\nicefrac{P}{R}$.
\end{definition}

\subsection{Compact automata}
\label{sec:compactness}

One key property when searching for lattice structures on PA is compactness.
Compactness definitions already have been introduced in \cite{panangaden:10} for the 
alternating model. 
\begin{definition}[adapted from Def.~9 in \cite{panangaden:10}]
\label{def:metric}
Let $P=\aPA$ be a PA. The function $d$ on $Dist(S)\times Dist(S)$ is defined by
$$d(\mu_1, \mu_2):= sup_{A\subseteq S}|\mu_1(A)-\mu_2(A)|$$
\end{definition}


Even if in \cite{panangaden:10} it is mentioned without a proof, that function $d$ defined above is really a metric, we give an explicit proof here.

\begin{lemma}
\label{lem:metric}
The function $d$ in Def.~\ref{def:metric} is a metric on $Dist(S)$.
\end{lemma}

\begin{proof}
It is clear by definition that $d$ is non-negative and symmetric.
Identity of indiscernibles is worth a thought: For general measures we only know when $d(\mu_1, \mu_2)=0$ that $\mu_1$ and $\mu_2$ coincide up to a 
zero set. As $Dist(S)$ is a set of \emph{discrete} probability measures, the values $\{(s,\mu(s))|s\in S\}$ completely determine $\mu$. Therefore
$d(\mu_1, \mu_2)=0 \Rightarrow \mu_1=\mu_2$.
Also the triangle inequality holds: $d(\mu_1, \mu_3)=sup_{A\subseteq S}|\mu_1(A)-\mu_3(A)|=sup_{A\subseteq S}|\mu_1(A)-\mu_2(A)+\mu_2(A)-\mu_3(A)|
\leq sup_{A\subseteq S}(|\mu_1(A)-\mu_2(A)|+|\mu_2(A)-\mu_3(A)|) \leq sup_{A\subseteq S}|\mu_1(A)-\mu_2(A)|+sup_{A\subseteq S}|\mu_2(A)-\mu_3(A)|=d(\mu_1, \mu_2)+d(\mu_2, \mu_3)$
\end{proof}

\begin{corollary}
Let $P=\aPA$ be a PA and fix an equivalence relation $R\in \{\sim, \approx\}$ on $S$, $s\in S$ and $a\in Act$. Then for
every pair $(s,a)\in S\times Act$ there is a metric space $(S_R(s,a),d)$.
\end{corollary}


\begin{definition}[adapted from Def.~10 in \cite{panangaden:10}]
\label{def:compactness}
Let $P=\aPA$ be a PA and consider some fixed equivalence relation $R\in \{\sim, \approx\}$ on $S$, $s\in S$ and $a\in Act$. 
We say that state $s$ is a-compact wrt.~$R$, if the set $S_R(s,a)$
%
%
is compact under the metric $d$. $P$ is called \emph{compact} wrt.~$R$, if $\forall s\in S \; \forall a\in Act: s$ is $a$-$compact$ wrt.~$R$.
We omit ``wrt.~$R$'' if the context is clear.
\end{definition}

The definition given by Desharnais et al.~can be seen also in the view of products of metric spaces. We recall the
definition of a product of metric spaces:

\begin{definition}
\label{def:product}
If $(M_i,d_i)$, $i\in \mathbb{N}$ are metric spaces, we define a metric on the countably infinite Cartesian product $\prod_{i\in \mathbb{N}}M_i$ by 
$$d(x,y)=\sum_{i=1}^\infty \frac{1}{2^i}\frac{d_i(x_i,y_i)}{1+d_i(x_i,y_i)}.$$
It is well-known that this construction leads to the metric space 
$(\prod_{i\in \mathbb{N}}M_i, d)$ \todo{add a citation here!}.
\end{definition}
This metric metrizes the product topology - under this metric all projection functions are continuous. 

\ms{Man muesste deutlicher machen, dass mit $d$ die Metrik aus Def.12 gemeint ist}
\js{Habe die fruehere Def in ein Korollar zu Lemma \ref{lem:metric} gewandelt. Dort ist der Kontext klar, denke ich. Uebrig geblieben ist nur der
Satz mit dem Produkt.}
We are now able to relate the metric space $(\prod_{(s,a)\in S\times Act} S_R(s,a), d)$ to the definition of compactness given in Def.~\ref{def:compactness}.

\begin{lemma}[Reformulation of Def.~\ref{def:compactness}]
Let $P=\aPA$ be a quotient of a PA wrt.~some equivalence relation $R$. 
The compactness of $P$ defined in Def.~\ref{def:compactness} is equivalent to the compactness of the 
metric space $(\prod_{(s,a)\in S\times Act} S_R(s,a), d)$.
\end{lemma}

\begin{proof}
\todo{make sure that this works on our topology. The usual product topology is defined by sup}
Assume that Def.~\ref{def:compactness} holds.
Tychonoff's theorem (or equivalently the axiom of choice) ensures that the product
of compact spaces is compact again, so the product space in Def.~\ref{def:product} is compact.

The projections $\prod_{(s,a)\in S\times Act} S_R(s,a)\rightarrow S_R(s,a)$ are continuous for every pair $(s,a)$.
Continuous mappings between metric spaces map compact sets to compact sets. So Def.~\ref{def:compactness} also holds.
\end{proof}

The question about compactness is only relevant in the case of infinite automata.
As long as PA are constructed out of finite sets (both states and transitions),
they are compact anyway. This is due to the following classical theorem:
\begin{theorem}[Heine-Borel]
For a subset of $M\subseteq\mathbb{R}^n$ ($\mathbb{R}^n$ the metric space of all n-Tuples with Euclidean metric) the following
statements are equivalent:
\begin{itemize}
  \item $M$ is bounded and closed
  \item every open cover of $M$ has a finite subcover, that is, $M$ is compact
\end{itemize}
\end{theorem}

An important theorem for compact sets is the Krein-Milman theorem (see below). It only works for locally convex vector spaces. 
We use as a basic fact, that the metric space $(l^\infty$, $d(x,y)=sup_{i\in \N}|x_i-y_i|)$ of bounded sequences in $\R$ is locally convex (using the seminorms $p_i(\{x_n\}_n)=|x_i|$, $i\in \N$).
Sequences may be used to characterise distributions in the following way. Assume that the states are ordered by the natural numbers. Define the mapping
$f:(S_R(s,a),d)\rightarrow (l^\infty, d)$, $\mu\mapsto (\mu(x_1),\mu(x_2),\ldots)$. We show that this mapping is continuous, which is straight-forward,
as singletons are also subsets.
$sup_{A\subseteq S}|\mu_1(A)-\mu_2(A)|<\epsilon \Rightarrow sup_{a\in S}|\mu_1(a)-\mu_2(a)|<\epsilon$.
As continuous images of compact sets are compact, we know that $f(S_R(s,a))$ is compact.
By definition of convex schedulers it is also clear that $f(S_R(s,a))$ is convex. 
The product of $\prod_{(s,a): S_R(s,a)\neq \emptyset} f(S_R(s,a))$ is also compact (Tychonoff's Theorem -- accepting the axiom of choice) and convex
as a product of convex sets. 
These sets will be used later to apply the Krein-Milman theorem.

\begin{definition}[Extreme points]
\label{def:extremal_points}
An \emph{extreme point} of a convex set $A$ is a point $x\in A$ with the property that if $x=cy+(1-c)z$ with $y,z \in A$ and $c \in [0,1]$,
then $y=x$ or $z=x$. $\extremalp (A)$ will denote the set of extreme points of $A$.
\end{definition}

%

Now the Krein-Milman theorem says that a compact convex subset of a locally convex vector space is the convex hull
of its extreme points:

\begin{theorem}[Krein-Milman \cite{Krein1940}]
\label{th:krein}
Let $A$ be a compact convex subset of a locally convex vector space $X$, then $A=CHull(\extremalp (A))$.
\end{theorem}


\section{Operations on bisimilar quotient automata}
\label{sec:operations}

\begin{definition}[Set operations on bisimilar quotients]
\label{def:setoperations}
Let $P=\aPA$ and $P'=\anotherPA$ be PA. 
Let $\stateSetDisjunct:=\stateSet \disjunctunion \stateSet'$ be the disjoint union of state spaces
and consider a fixed bisimulation relation $R\in \{\sim, \approx\}$, $R\subseteq \stateSetDisjunct \times \stateSetDisjunct$.
Let $P$ and $P'$ be bisimilar, i.e.~$\startState R \startState'$. In the following, we assume that $P$ and $P'$ are
quotients wrt.~$R$.
The intersection of $S$ and $S'$ (and the union) is performed in 
$\nicefrac{\stateSetDisjunct}{R}$ by means of the canonical mappings
$$
\xymatrix{
\stateSet \ar[r] \ar@{-->}[rd]          & \stateSetDisjunct \ar[d] & \stateSet' \ar[l] \ar@{-->}[ld]      \\
 & \nicefrac{\stateSetDisjunct}{R}  & 
}
$$
With this diagram in mind we may define
$$S'\cap_R S = \{[x] \in \nicefrac{\stateSetDisjunct}{R} | \exists s\in S:s\in [x] \wedge \exists s' \in S':s' \in [x]\}$$
$$S'\cup_R S = \{[x] \in \nicefrac{\stateSetDisjunct}{R} | \exists s\in S:s\in [x] \vee \exists s' \in S':s' \in [x]\}$$
Using these canonical mappings, similar operations can be defined on the sets of transitions. Clearly $\transitionRelation \subseteq \stateSet \times Act \times \Disc{\stateSet}$ and
$\transitionRelationB \subseteq \stateSet' \times Act \times \Disc{\stateSet'}$. 
The set operations are performed in the set 
$\nicefrac{\stateSetDisjunct}{R} \times Act \times \Disc{\nicefrac{\stateSetDisjunct}{R}}$ by means of the corresponding canonical mappings.
$$
\xymatrix{
\stateSet \times Act \times \Disc{\stateSet} \ar[r] \ar@{-->}[rd]          & \stateSetDisjunct \times Act \times \Disc{\stateSetDisjunct} \ar[d] & \stateSet' \times Act \times \Disc{\stateSet'} \ar[l] \ar@{-->}[ld]      \\
 & \nicefrac{\stateSetDisjunct}{R} \times Act \times \Disc{\nicefrac{\stateSetDisjunct}{R}}  & 
}
$$
We thus may define (using the abbreviation $\mathfrak{T}=\nicefrac{\stateSetDisjunct}{R} \times Act \times \Disc{\nicefrac{\stateSetDisjunct}{R}}$):
$$\transitionRelation'\cap_R \transitionRelation = \{([x],a,[\mu]) \in \mathfrak{T} | (\exists (s,a,\gamma) \in \transitionRelation: s\in [x] \wedge \gamma\in [\mu]) \wedge (\exists (s',a,\gamma') \in \transitionRelation': s'\in [x] \wedge \gamma'\in [\mu])\}$$
$$\transitionRelation'\cup_R \transitionRelation = \{([x],a,[\mu]) \in \mathfrak{T} | (\exists (s,a,\gamma) \in \transitionRelation: s\in [x] \wedge \gamma\in [\mu]) \vee (\exists (s',a,\gamma') \in \transitionRelation': s'\in [x] \wedge \gamma'\in [\mu])\}$$
%
\ms{Ist die Bedeutung von $[\mu]$ in den beiden vorigen Zeilen klar?}
\js{Koennte man noch definieren, ich denke, es ist genug Platz}
Finally we note that $\startState$ has to be mapped to $\nicefrac{\stateSetDisjunct}{R}$.
Summing up, we define the intersection quotient wrt.~$R$ as
$$P \cap P':=(S'\cap_R S, Act, \transitionRelation \cap_R \transitionRelation', [\startState]_R)$$
and, similarly, the union quotient wrt.~$R$ as 
$$P \cup P':=(S'\cup_R S, Act, \transitionRelation \cup_R \transitionRelation', [\startState]_R)$$
With the same mappings one may define $P \subseteq P'$ if and only if $S \subseteq S'$
 (in $\nicefrac{\stateSetDisjunct}{R}$) and $\transitionRelation \subseteq \transitionRelationB$
in $\nicefrac{\stateSetDisjunct}{R} \times Act \times \Disc{\nicefrac{\stateSetDisjunct}{R}}$.
\end{definition}

As we target on bisimilar automata in this paper, we would like to stress that only the reachable fraction of automata is relevant.
Our quotient definition and the above set operations are defined for the general case, where unreachable states may occur, but of course
they apply also to the case where only reachable states are present.

\begin{lemma}
\label{lem:isocup}
Let $R\in \{\sim, \approx\}$ and
$P=\aPA$, $P'=\anotherPA$ be bisimilar quotients of PA wrt.~$R$.
It holds that $P \cap P' \someequiv P' \cap P$ and $P \cup P' \someequiv P' \cup P$.
\end{lemma}

\begin{proof}
Clear by identifying the different direct sums according to $R$
\end{proof}

As it will turn out that $\tau$-loops leading back to the same state
(as part of a distribution) can disturb the lattice property in the
case of weak bisimulation, we give the following definition.

\begin{definition}[Rescaled Automata]
\label{def:rescale}
An automaton $P=\aPA$ is called \emph{rescaled}, if for all $(s,\tau,\mu)\in \transitionRelation$
it holds either that $\mu(s)=0$ or $\mu(s)=1$.
\end{definition}
\ms{Folgender Satz ist neu (soll JS-Kommentar zusammenfassen):}
Note that the rescaling as defined in Def.~\ref{def:rescale}
is only one possibility (incidently the one leading to the smallest transition
fanout), but different rescalings, where each $\tau$ transition would
loop back to its source state with fixed probability $0\leq p < 1$ (or probability $1$ for ``loops'') 
would also work.
\js{Achtung: Die Lattice-Geschichte laeuft fuer beliebige Rescalings. Wenn Du fuer alle $\tau$-Transitionen in Deinen PA-Mengen abmachst, 
dass immer mit W-keit $p$ zurueck zum Ursprung
gegangen wird, klappt auch das Schneiden! Ich habe halt p=0 (rescaled) gewaehlt. Beliebige andere $p\neq 1$ funktionieren aber auch.
Nur \emph{unterschiedliche} Rescalings zu schneiden funktioniert nicht. Merke: Lattices sind immer moeglich, sobald das Rescaling-$p$ fest gewaehlt ist!
Die rescaled Automaten mit $p=0$ zeichnen sich halt unter all diesen Moeglichkeiten dadurch aus, dass sie minimalen Fanout haben.}

\begin{remark}
Let $P=\aPA$ be a PA. We split its set of transitions into the sets
$T_{\neg\tau}:=\{(s,a,\mu) \in T | a\neq \tau \}$, $T_{\Delta}=\{(s,\tau,\mu) \in T | \mu=\Delta_s \}$ and
$T_{\neg\Delta}=\{(s,\tau,\mu) \in T | \mu \neq \Delta_s \}$.
Now we define the function $res: T_{\neg\Delta} \rightarrow S\times Act \times Dist(S)$ by
$res(s,\tau,\mu):=(s,\tau,\frac{1}{1-\nu(s)}(\nu-s))$
which is well-defined, as we always have $\nu(s)\neq 1$. With 
$$T':=T_{\neg\tau} \cup T_{\Delta} \cup res(T_{\neg\Delta})$$
we may define the rescaled automaton $P^{res}:=(S,Act,T',s_0)$.
When using 
randomised 
schedulers, it is a basic fact, that $P\approx P^{res}$.
\ms{Ist das wirklich so ein basic fact?}
\js{nun ja. Aus der reskalierten Transition bekommt man ohne Probleme durch Hinzunahme des Schedulers, der in $s$ bleibt, zum nicht reskalierten Fall.
Umgekehrt erreicht man aus dem nicht-reskalierten Fall durch den DD-Scheduler, der aus $s$ immer die nicht-reskalierte Transition nimmt, sicher den Fall,
bei dem $s$ mit W-keit 0 erreicht wird, also den reskalierten Fall.}
\end{remark}

\subsection{Sets of quotients}

\begin{definition}[Sets of quotients]
\label{def:SetsOfQuotients}
Let P be an automaton and $\mathcal{PA}$ be the set of all PA. Define the set
$$\mathfrak{Q}_\sim (P):=\{ A \in \mathcal{PA} | A\text{ quotient wrt.~}\sim, A\sim P \}$$
and the set 
$$\mathfrak{Q}_\approx (P):=\{ A \in \mathcal{PA} | A\text{ quotient wrt.~}\approx, A\approx P, A \text{ rescaled} \}$$
The reachable fractions can also be considered:
$$\mathfrak{Q}^*_\sim (P):=\{ A \in \mathcal{PA} | A\text{ quotient wrt.~}\sim, A\sim P, A=r(A) \}$$
and the set 
$$\mathfrak{Q}^*_\approx (P):=\{ A \in \mathcal{PA} | A\text{ quotient wrt.~}\approx, A\approx P, A \text{ rescaled}, A=r(A) \}$$
\end{definition}

Of course, the automata in the set $\mathfrak{Q}^*_\sim (P)$ (in the set $\mathfrak{Q}^*_\approx (P)$) 
all have the same number of states.

\begin{lemma}
It holds that $\mathfrak{Q}^*_\sim (P) \subseteq \mathfrak{Q}_\sim (P)$ and $\mathfrak{Q}^*_\approx (P) \subseteq \mathfrak{Q}_\approx (P)$.
\end{lemma}

\begin{remark}
As usual for equivalence relations, we may consider quotients
wrt.~isomorphism, i.e.~$\nicefrac{\mathfrak{Q}_\sim (P)}{\someequiv}$ 
and $\nicefrac{\mathfrak{Q}_\approx (P)}{\someequiv}$.
Without loss of generality we may identify the states of
every automaton in these quotients by a subset of the natural numbers 
$\mathbb{N}$ (using $1$ for the initial states).
For the sets $\mathfrak{Q}_\sim (P)$ and $\mathfrak{Q}_\approx (P)$ such an enumeration 
cannot be performed, as there are uncountably many unreachable parts (this is shown in detail by Lemma \ref{lemma:uncountable}).
\end{remark}

For the rest of the paper we now assume that the states in $\nicefrac{\mathfrak{Q}^*_\sim (P)}{\someequiv}$ 
and $\nicefrac{\mathfrak{Q}^*_\approx (P)}{\someequiv}$ are consecutive natural numbers,
starting with 1 for the initial state.


The rest of this section is devoted to show that for compact $P$ there are well-defined operations
$$\cap: \nicefrac{\mathfrak{Q}_\sim (P)}{\someequiv} \times \nicefrac{\mathfrak{Q}_\sim (P)}{\someequiv} \rightarrow \nicefrac{\mathfrak{Q}_\sim (P)}{\someequiv}$$
$$\cap: \nicefrac{\mathfrak{Q}_\approx (P)}{\someequiv} \times \nicefrac{\mathfrak{Q}_\approx (P)}{\someequiv} \rightarrow \nicefrac{\mathfrak{Q}_\approx (P)}{\someequiv}$$
Similar operations exist for $\cup$, but their existence and well-definedness is clear.
Note that it is a priori \emph{not} clear that the result is again in $ \nicefrac{\mathfrak{Q}_\sim (P)}{\someequiv}$ (or $\nicefrac{\mathfrak{Q}_\approx (P)}{\someequiv}$, respectively),
with other words that the results are again bisimilar to $P$. This will be shown in the following.



\begin{remark}[Counterexample]
\label{rem:counterexample}
We show that there is in general no well-defined operation $\cap: \nicefrac{\mathfrak{Q}_\sim (P)}{\someequiv} \times \nicefrac{\mathfrak{Q}_\sim (P)}{\someequiv} \rightarrow \nicefrac{\mathfrak{Q}_\sim (P)}{\someequiv}$,
that means the compactness of $P$ is crucial. We use the (non-compact) automata
$$P:=(\{1,2,3\}, \{\tau , a, b \}, \{(1,\tau,\frac{1}{n}\Delta_2 \oplus (1-\frac{1}{n})\Delta_3)\}_{n\in \mathbb{N}, n\geq 2} \cup \{(2,a,\Delta_2),(3,b,\Delta_3)\},1 )$$
and
$$P':=(\{1,2,3\}, \{\tau , a, b \}, \{(1,\tau,\frac{\mathrm e}{n}\Delta_2 \oplus (1-\frac{\mathrm e}{n})\Delta_3)\}_{n\in \mathbb{N}, n\geq 6}$$ 
                 $$ \cup \{(1,\tau,\frac{1}{2}\Delta_2 \oplus \frac{1}{2}\Delta_3),(2,a,\Delta_2),(3,b,\Delta_3)\},1 ).$$
It is easy to see that $P\sim P'$ and $P,P'\in \nicefrac{\mathfrak{Q}_\sim (P)}{\someequiv}$.
It is clear that the intersection of both automata is only 
$$P^\cap:=(\{1,2,3\}, \{\tau , a, b \},\{(1,\tau,\frac{1}{2}\Delta_2 \oplus \frac{1}{2}\Delta_3),(2,a,\Delta_2),(3,b,\Delta_3)\},1 ),$$ as $\mathrm e$ is \emph{not} a rational number.
So the result of the intersection is clearly not bisimilar, i.e.~not in $\nicefrac{\mathfrak{Q}_\sim (P)}{\someequiv}$.
The reason for this is, that both automata allow for a limiting distribution $(1,\tau, \Delta_3)$, but both don't reach this distribution.
As the ways how this limit is reached (i.e.~the sequences) are different, the intersection no longer leads to a bisimilar result.
Note that there are many other examples (e.g.~every irrational root may be taken instead of $\mathrm{e}$).
\end{remark}

Thus we have shown:
\begin{lemma}
In general
$(\nicefrac{\mathfrak{Q}_\sim (P)}{\someequiv},\subseteq)$ and $(\nicefrac{\mathfrak{Q}^*_\sim (P)}{\someequiv},\subseteq)$ are \emph{not} lattices.
\end{lemma}

\section{Lattice structures}
\label{sec:structures}
For the rest of this paper we consider a subset of PA where $(\nicefrac{\mathfrak{Q}_R (P)}{\someequiv},\subseteq)$ and $(\nicefrac{\mathfrak{Q}^*_R (P)}{\someequiv},\subseteq)$,
$R\in\{\sim,\approx\}$ 
are lattices.

%
\begin{lemma}[compact sets of quotients]
Let $R\in \{\sim, \approx\}$.
The PAs contained in $\nicefrac{\mathfrak{Q}_R (P)}{\someequiv}$ are either
all compact, or all non-compact. Thus
the property of compactness is well-defined for $\nicefrac{\mathfrak{Q}_R (P)}{\someequiv}$.
\end{lemma}

\begin{proof}
Let $P$ be a PA, $R\in \{\sim, \approx\}$ and $P'\in \nicefrac{ \mathfrak{Q}_R (P)}{\someequiv}$.
$P'$ is compact if and only if $P$ is compact, as the sets $S_R(s,a)$ are unique up to isomorphism. 
\end{proof}

In the following we will only consider compact PA and therefore also compact sets of quotients.
%
%
%
In metric spaces, compactness is equivalent to sequence compactness.

\begin{lemma}
\label{lem:A}
$(\nicefrac{\mathfrak{Q}_\sim (P)}{\someequiv},\subseteq)$ is a poset. 
If $P$ is compact, it is even
a lattice with union (intersection) of automata as join (meet) operations. Intersections are also possible over arbitrary sets of bisimilar 
automata.
\end{lemma}

\begin{proof}

By Krein-Milman Theorem (Theorem \ref{th:krein}) we know that for countable state spaces the sets $S_\sim(s,a)$ (seen as subsets of $l^\infty$)
have a unique set of extreme points $\extremalp(S_\sim(s,a))$. This set is included in every bisimilar automaton, therefore in every intersection.
This is the reason why the intersection of all automata in a set of bisimilar automata still leads to a bisimilar automaton.
Transitions leading to some distribution in $S_\sim(s,a)$ that is not extreme do not change the bisimilarity. This is the reason why the union
of two bisimilar automata leads to a bisimilar automaton.
Finally observe that the unreachable fraction of states and transitions does not play any role for
bisimilarity.
\end{proof}


In the strong case, a minimal set of transitions leading to extreme points can be chosen from the strong transitions emanating 
from $s$. In the weak case this does not have to be the case, as an extreme point might be reached transitively via some intermediate distributions.

Astonishingly, there's a similar statement also for weak bisimilarity, when only rescaled automata are considered.
The interesting point is that the arguments of the proof of Lemma~\ref{lem:A} do not apply one-to-one, as the extreme points 
can now also be weak transitions. Therefore it is (with the arguments of the proof for strong bisimulation) not clear 
that the intersection leads again to a weakly bisimilar PA.

\begin{lemma}
\label{lem:B}
$(\nicefrac{\mathfrak{Q}_\approx (P)}{\someequiv}, \subseteq)$ is a poset.
When $\nicefrac{\mathfrak{Q}_\approx (P)}{\someequiv}$ is compact,
it is even a lattice with union (intersection) of automata as join (meet) operations.
Intersections are also possible over arbitrary sets of bisimilar 
automata.
\end{lemma}

\begin{proof}
The meet operation for \emph{finite} PA is justified by Lemma 12 
in \cite{tacas:13} 
and the observation that the unreachable fraction of states and transitions does not play any role for
bisimilarity. The join operation is clear by definition of weak bisimilarity and the same observation.

For the infinite case assume for the rest of the proof that there are two quotients $P_1$ and $P_2$, $P_1\approx P_2$ 
and we identify the (countable) reachable state spaces $S_1$ and $S_2$ 
as in Def.~\ref{def:setoperations} by $\nicefrac{S^\circ}{R}$ (note that this induces a bijection). 


Now fix one state $s \in \nicefrac{S^\circ}{\approx}$.
 By assumption it must have two sets of emanating $a$-transitions -- one in $P_1$, one in $P_2$.
During the proof we use the notation $S^i_R(s,a)$, $i\in\{1,2\}$ to denote the sets of reachable distributions of automaton 1 and 2, respectively. 
We have to show that already the intersection of both sets is enough to generate $S_\approx(s,a)$. 
The main difference to the strong case is that we only know that
$S^1_\approx(s,a)=S^2_\approx(s,a)$ for all $a\in Act$, but this does not necessarily imply that also $S^1_\sim(s,a)=S^2_\sim(s,a)$.
According to Def.~\ref{def:SetsOfQuotients} both automata $P_1$ and $P_2$ are rescaled.
We consider the one-step-transitions $T^i_a=\{\mu \in Dist(\nicefrac{S^\circ}{R}) | (s,a,\mu)\in T_i \}$.
Assume further that the identical transitions (modulo bijection) already have been identified (in $\nicefrac{S^\circ}{R}\times Act\times Dist(\nicefrac{S^\circ}{R})$) and are denoted 
$T^\cap_a = T^1_a \cap T^2_a$ and let $T^{\neg\cap_i}_a=T_a^i \setminus T^\cap_a$.
We will show that the identical transitions are enough to cover all the behaviour, i.e.~all transitions in 
$T^{\neg\cap_i}_a$, $i\in \{1,2\}$, can be omitted.

We will first consider the case $a=\tau$.
We fix an arbitrary $\mu\in T^{\neg\cap_1}_\tau$ for which we show that it can be
generated by transitions from $T^{\cap}_\tau$.
It is clear that both sets $T^1_\tau$ and $T^2_\tau$
generate the set $S_\approx(s,\tau)$ (in the sense that all weak rescaled transitions must start with a combination of these transitions as first step).

With the construction from the proof of Lemma 12 
in \cite{tacas:13} we get a 
weak
transition\footnote{The proof in \cite{tacas:13} shows that a transition $s\stackrel{\tau}{\rightarrow}\mu$ 
-- there denoted $s\stackrel{\tau}{\rightarrow}\nu_s$ -- 
is redundant, as long as it is not in the intersection $T^{\cap}_\tau$, by constructing
the above mentioned transition. The only case where it is not redundant leads to a contradiction to the quotient property.}
in $T_1\setminus \{(s,\tau,\mu)\}$ such that $s\Rightarrow_C \mu$. 
\ms{ok, nach Studium des Tacas-Beweises kaufe ich das ab.}
Note that the construction in \cite{tacas:13} has been given for \emph{finite} state spaces, but all constructions
there are also possible in the countable compact case due to sequential compactness.
We split this transition in its first-step-probability $\mu^{(1)}$, such that $s \rightarrow_C \mu^{(1)}\Rightarrow_C \mu$.
We see that without loss of generality we may reduce the set $T_1$ to the set $T_1\setminus \{(s,\tau,\mu)\}$ without losing bisimilarity.
With this construction we proceed as follows: Pick now some other transition $(s,\tau,\mu')$ with $\mu'\in T^{\neg\cap_1}_\tau\setminus \{\mu\}$
that is used somewhere in the weak transition $s \rightarrow_C \mu^{(1)}\Rightarrow_C \mu$.
\ms{Ok, ich habe das zurueckgeaendert auf "weak transition".
``somewhere'' bedeutet also entweder im $\rightarrow_C$-Teil oder im $\Rightarrow_C$-Teil der schwachen Transition!
Wenn $(s,\tau,\mu')$ im $\rightarrow_C$-Teil verwendet wuerde,
bekommen wir unten ein $\mu^{(2)} \not= \mu^{(1)}$.
Im anderen Fall, wenn $(s,\tau,\mu')$ im $\Rightarrow_C$-Teil verwendet wuerde,
bekommen wir eine Schleife $s \Rightarrow_C s$, dann tritt u.g. Fall $\mu^{(1)}=\mu^{(2)}$ ein.
}
\js{Man muss ja \emph{alle} nicht-schnitt-Transitionen ersetzen, die irgendwo in der
besagten schwachen Transition vorkommen (Stichwort Loops zurueck zu s). Wenn man sich nur auf den Anfangsteil stuerzt dann ist nicht gesagt, dass nicht 
irgendwo spaeter doch noch eine nicht-schnitt-Transition verwendet wird. Deshalb auch unten die ``possibly'' Betrachtung: Dort wird gerade im ``hinteren''
Teil der schwachen Transition was getauscht.}
\ms{Ich glaub' das hab ich jetzt kapiert}
\js{Schaut gut aus! Ich habe mich gefragt, ob man noch genauer hinschreiben sollte, dass \emph{alle} entsprechenden starken Uebergaenge in einem Schritt durch 
die entsprechenden schwachen Uebergaenge ersetzt werden muessen. Nach einem Ersetzungsschritt darf die transition $(s,\tau,\mu)$ \emph{nirgends} mehr verwendet werden.}
Now, again by the construction in the proof of Lemma 12 in \cite{tacas:13}, substitute all strong $\tau$ transitions 
from $s$ to $\mu'$ 
by weak transitions in $T_1\setminus \{(s,\tau,\mu) ,(s,\tau,\mu')\}$ such that $s \rightarrow_C \mu^{(2)}\Rightarrow_C\mu^{(1)}\Rightarrow_C \mu$
(possibly $\mu^{(2)}=\mu^{(1)}$ when we substituted some transition not taken in the first step).
Continuing in the same way yields a sequence of combined one-step-transitions leading to the series
of distributions
$(\mu^{(i)})_{i\in \mathbb{N}}$. 


Now compactness comes into play: This sequence must have limit points \emph{in} $S_\approx(s,\tau)$. 
If there is \js{Habe jetzt den einzelnen Punkt herausgenommen, da eventuell ``weiter hinten'' noch nicht-schnitt Transitionen auftauchen koennten.} 
one limit point already in $T^\cap_\tau$, 
by construction there will be a scheduler that uses only transitions from $T^\cap_\tau$ when leaving $s$. 
If there is a limit point $\mu^{\ast}$ in $T^{\neg\cap_1}_\tau$, there is a non-trivial loop $\mu^\ast \Rightarrow_c \gamma^\ast \Rightarrow_c \mu^\ast$
for some $\gamma^\ast \in T^{\neg\cap_2}_\tau$ which 
would contradict the quotient property as it would render at least two states bisimilar.
Therefore we conclude that this case cannot happen and $\mu$ can be omitted from $T^{\neg\cap_1}_\tau$ without losing bisimilarity.

As $\mu$ was chosen arbitrarily, this construction can be done for all elements in $T^{\neg\cap_1}_\tau$ -- even if they might be uncountably
many -- (and analogously for $T^{\neg\cap_2}_\tau$).
Therefore all transitions in $T^{\neg\cap_1}_\tau$ (and analogously for $T^{\neg\cap_2}_\tau$) may be omitted without losing bisimilarity.

For the case $a\neq \tau$ we may omit those strong $a$-transitions that can be mimicked by a weak a transition (after leaving out this strong transition).
Observe that there cannot be cyclic relations ``$s$ uses the $a$ transition of $t$ (with probability 1)'' and vice versa, 
because then $s$ and $t$ would be bisimilar, which
is a contradiction to the quotient property.

For the general case of an arbitrary intersection (i.e.\ possibly more than two automata)
observe the following:
Let $A$ be an arbitrary set of weakly bisimilar quotients. We have to show that the automaton $P_0=\cap_{P\in A}P$
is still bisimilar. Pick an arbitary automaton $P \in A$ we have to show that all transitions of $P$ that are not in $P_0$ are redundant.
We may find for every such transition $tr$ a bisimilar automaton $P_{tr}$ that does not have such a transition (for otherwise the transition would have to be
in the intersection).
By the above procedure for two automata we show that the transition $tr$ is redundant. This can be done for all other transitions not in $P_0$.
\end{proof}

\begin{remark}
The crucial point in the proof of Lemma \ref{lem:B} is compactness. 
For general quotients a sequence of redundant distributions may not have a limit in $S_\approx(s,\tau)$
(This would be e.g.~the distribution $(1,\tau,\Delta_3)$ in Remark \ref{rem:counterexample}). 
Sequential compactness ensures that there exists such a limit in  $S_\approx(s,\tau)$.
\end{remark}


As we have shown that the intersection of arbitrary sets of bisimilar automata is again bisimilar by 
Lemma \ref{lem:A} and \ref{lem:B}, the following theorem is immediate. It extends Theorem 1 of \cite{tacas:13} to the 
compact automata case.

\begin{theorem}
\label{th:A}
Let $P$ be compact wrt.~$R$. Then $(\nicefrac{\mathfrak{Q}_R (P)}{\someequiv},\subseteq)$ has a 
unique minimal element. 
\end{theorem}
\js{Das haette ich gerne separat ausgekostet: Fuer endliche Faelle wurde das im TACAS papier definiert und stimmt mit unserer Notation ueberein.
    Unsere Notation bietet aber die Einsicht ebenfalls fuer unendliche Automaten, die mit den Preorders im TACAS Papier ueberhaupt nicht erfasst werden koennen.}

%
%

\begin{definition}
\label{def:NF}
Let $P$ be compact wrt.~$R$, $R\in\{\sim, \approx\}$. Then there is a well-defined mapping 
$\mathcal{N}: PA \rightarrow \nicefrac{PA}{\someequiv}$, given by $P \mapsto \bot(\nicefrac{\mathfrak{Q}_R (P)}{\someequiv})$ 
that assign to every PA $P$ the minimal element in $\nicefrac{\mathfrak{Q}_R (P)}{\someequiv}$.
The mapping to $\nicefrac{\mathfrak{Q}_\sim (P)}{\someequiv}$ is just quotienting, while the mapping to $\nicefrac{\mathfrak{Q}_\approx (P)}{\someequiv}$
is quotienting followed by rescaling.
The mapping $\mathcal{N}$ (for $R$ fixed) is called \emph{normal form}.
\end{definition}

This definition corresponds to the normal form definition given in \cite{tacas:13}.

\begin{corollary}
The notations for normal forms defined in Def.~\ref{def:NF} are normal forms in the sense of \cite{tacas:13}.
\end{corollary}


As the bisimilarity only considers the reachable fraction of the state space, the following corollary is immediate.
\begin{corollary}
Lemma \ref{lem:A}, Lemma \ref{lem:B} and Theorem \ref{th:A} also hold for $\nicefrac{\mathfrak{Q}^*_\sim (P)}{\someequiv}$
instead of $\nicefrac{\mathfrak{Q}_\sim (P)}{\someequiv}$ and $\nicefrac{\mathfrak{Q}^*_\approx (P)}{\someequiv}$ instead of 
$\nicefrac{\mathfrak{Q}_\approx (P)}{\someequiv}$.
\end{corollary}

Some of the lattices we have constructed above have the appealing property of being bounded.

\begin{lemma}
The lattices $(\nicefrac{\mathfrak{Q}^*_\sim (P)}{\someequiv},\subseteq)$ and $(\nicefrac{\mathfrak{Q}^*_\approx (P)}{\someequiv},\subseteq)$
are bounded.
\end{lemma}

\begin{proof}
The lower bound is clearly given by the normal forms defined above. So it remains to show that there is also an upper bound.
Let $P=\aPA$, for $(\nicefrac{\mathfrak{Q}^*_\sim (P)}{\someequiv},\subseteq)$ and $(\nicefrac{\mathfrak{Q}^*_\approx (P)}{\someequiv},\subseteq)$ 
we know that it is sufficient to consider $(\nicefrac{\mathfrak{Q}^*_\sim (r(P))}{\someequiv},\subseteq)$ and $(\nicefrac{\mathfrak{Q}^*_\approx (r(P))}{\someequiv},\subseteq)$.
As $\stateSet$ is countable, also the set of reachable states will be countable.
As there is no restriction on $\transitionRelation$ wrt.~finiteness, the union over all transition relations in $\nicefrac{\mathfrak{Q}^*_R (P)}{\someequiv}$ (identified by
R) will be a valid transition relation (i.e.~$\cup_R$ of all transition relations of quotients in the set $\nicefrac{\mathfrak{Q}^*_R (P)}{\someequiv}$).
\end{proof}
\ms{Ok, ich denke ich verstehe das. Der Punkt ist, dass jetzt Automaten
mit (ueberabzaehlbar) unendlich vielen Transitionen erlaubt sind. Unendlichkeit
des Zustandsraums bzw der Aktionsmenge braucht man gar nicht.}

\begin{lemma}
For a finite automaton $P$,
\begin{enumerate}
  \item the maximal elements in $(\nicefrac{\mathfrak{Q}^*_\sim (P)}{\someequiv},\subseteq)$ and $(\nicefrac{\mathfrak{Q}^*_\approx (P)}{\someequiv},\subseteq)$ are not necessarily finite.
  \item the minimal elements in $(\nicefrac{\mathfrak{Q}^*_\sim (P)}{\someequiv},\subseteq)$ and $(\nicefrac{\mathfrak{Q}^*_\approx (P)}{\someequiv},\subseteq)$ are finite.
\end{enumerate}
\end{lemma}

\begin{proof}
Assume that $P=\aPA$ is given where $\stateSet=\{A,B,C\}$. Assume further that transitions $(A,\tau,\Delta_B)$ and $(A,\tau,\Delta_C)$
exist. Then every $(A,\tau,c\Delta_B\oplus (1-c)\Delta_C)$ would also be a bisimilar transition for every $c\in (0,1)$. 
Thus, the number of such distributions is uncountably infinite.
The union of all those transitions would therefore lead to a non-finite automaton.
\end{proof}

\begin{corollary}
For a finite automaton $P$,
the minimal elements in $(\nicefrac{\mathfrak{Q}_\sim (P)}{\someequiv},\subseteq)$ and $(\nicefrac{\mathfrak{Q}_\approx (P)}{\someequiv},\subseteq)$ are finite.
\end{corollary}

\begin{lemma}
\label{lemma:uncountable}
The lattices $(\nicefrac{\mathfrak{Q}_\sim (P)}{\someequiv},\subseteq)$ and $(\nicefrac{\mathfrak{Q}_\approx (P)}{\someequiv},\subseteq)$ are unbounded
whenever there is at least one action $a\in Act$, $a\neq \tau$.
\end{lemma}

\begin{proof}
\js{Beweis (und Formulierung des Lemmas) ziemlich neu, daher auch fehleranfaellig...}
Let $P=\aPA$ and $a\in Act$, $a\neq \tau$. We can construct two unreachable non-bisimilar states $s_1$, 
$s_2$ by using the transition $s_1\stackrel{a}{\rightarrow}\Delta_{s_1}$ which is not possible by $s_2$
(possibly $s_1\sim s$ or $s_2\sim s$ for some $s\in S$, then we would use the corresponding state(s) from $S$,
not adding them by a disjunct union in the following construction).
So without loss of generality we may assume that $P=(S\disjunctunion \{s_1,s_2\}, Act, T \cup \{s_1\stackrel{a}{\rightarrow}\Delta_{s_1}\})$
The further construction goes in two steps. The first step is to construct a countable set of `fresh' distinguished unreachable states
$U=\mathbb{N}$ in the following way.
use a single transition $i\stackrel{\tau}{\rightarrow}\frac{1}{i}\Delta_{s_0}\oplus (1-\frac{1}{i})\Delta_{s_1}$ for every $i\in U$
(again, if $i\sim s$ for some $s\in S$, we may assume that $s\in U$, but will not add $s$ again to the state space). So we reach the
bisimilar automaton $P'=(S\disjunctunion U, Act, T \cup \{i\stackrel{\tau}{\rightarrow}\frac{1}{i}\Delta_{s_0}\oplus (1-\frac{1}{i})\Delta_{s_1}\}_{i\in U}, s_0)$.
It is easy to see that all states in $U$ are not bisimilar. 
For the second step notice that the powerset of a countable is uncountable.
For every such subset $A\subseteq \mathbb{N}$ we may choose a distribution $\mu_A$ where $\mu_A(s)>0$ if $s\in A$, $0$ otherwise.
Now, we can construct (uncountably many) PA $P'_A$ bisimilar to $P$ where we add one additional unreachable state $s_A$ with the transition 
$s_A\stackrel{\tau}{\rightarrow}\mu_A$.
Each of those unreachable states would need a separate state in the maximal element of our lattice (leaving alone the states that are bisimilar
to one state of $S$ by chance), 
meaning that there would be an uncountable number of states
in the maximal element. This is a contradiction to the countability of the state space.
The weak case follows similarly.
\end{proof}

\section{Examples}
\label{sec:examples}

The first example is given in Fig.~\ref{fig:strong_bisim}.
We start with the two strongly bisimilar automata at the top of the figure.
The intersection of both is the automaton at the bottom of the figure.\\

\begin{figure}
\begin{center}
  \subfloat[Strongly bisimilar case]{\label{fig:strong_bisim}\includegraphics[width=5cm]{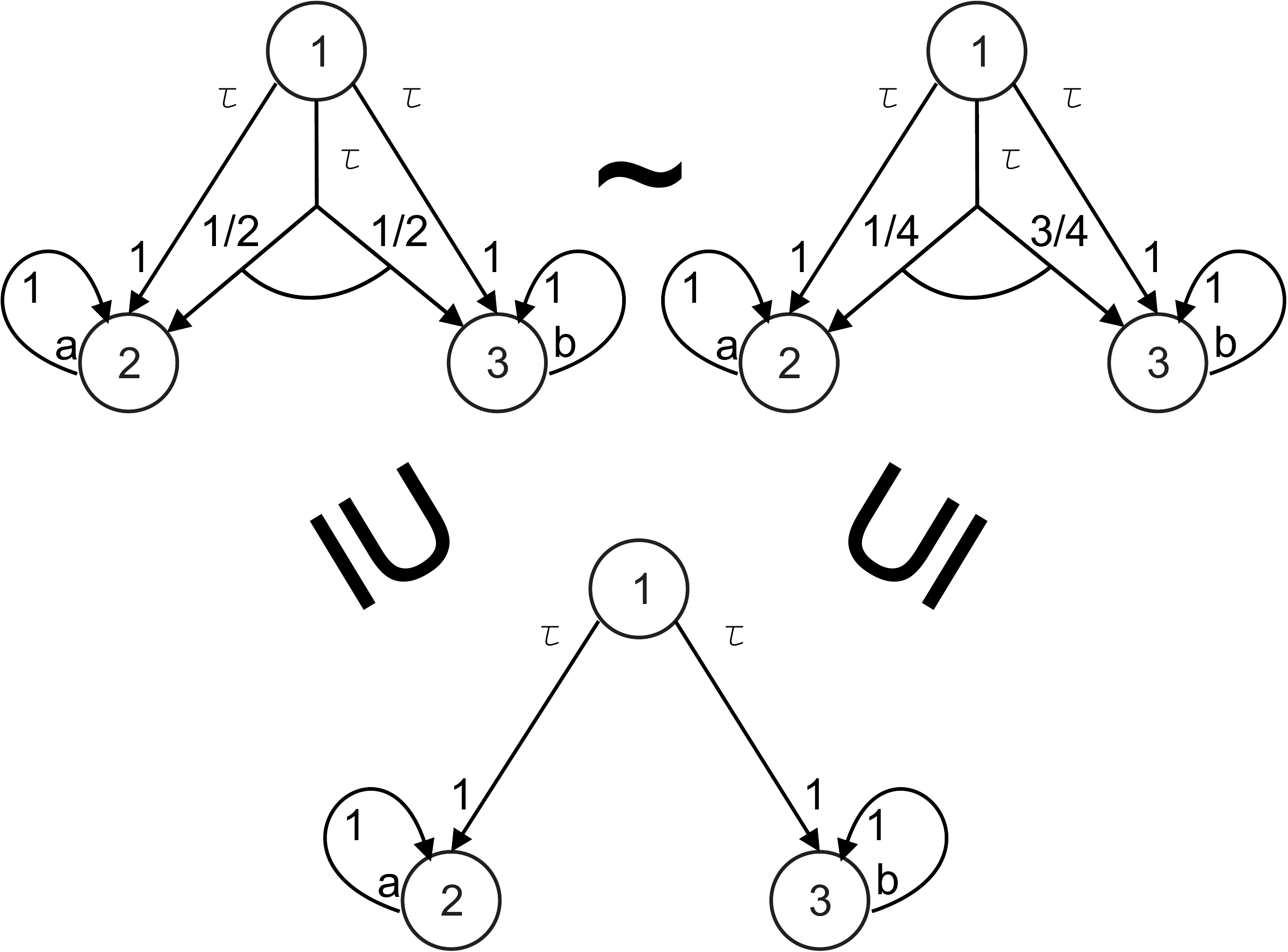}} \qquad
  \subfloat[Weakly bisimilar case]{\label{fig:weak_bisim}\includegraphics[width=5cm]{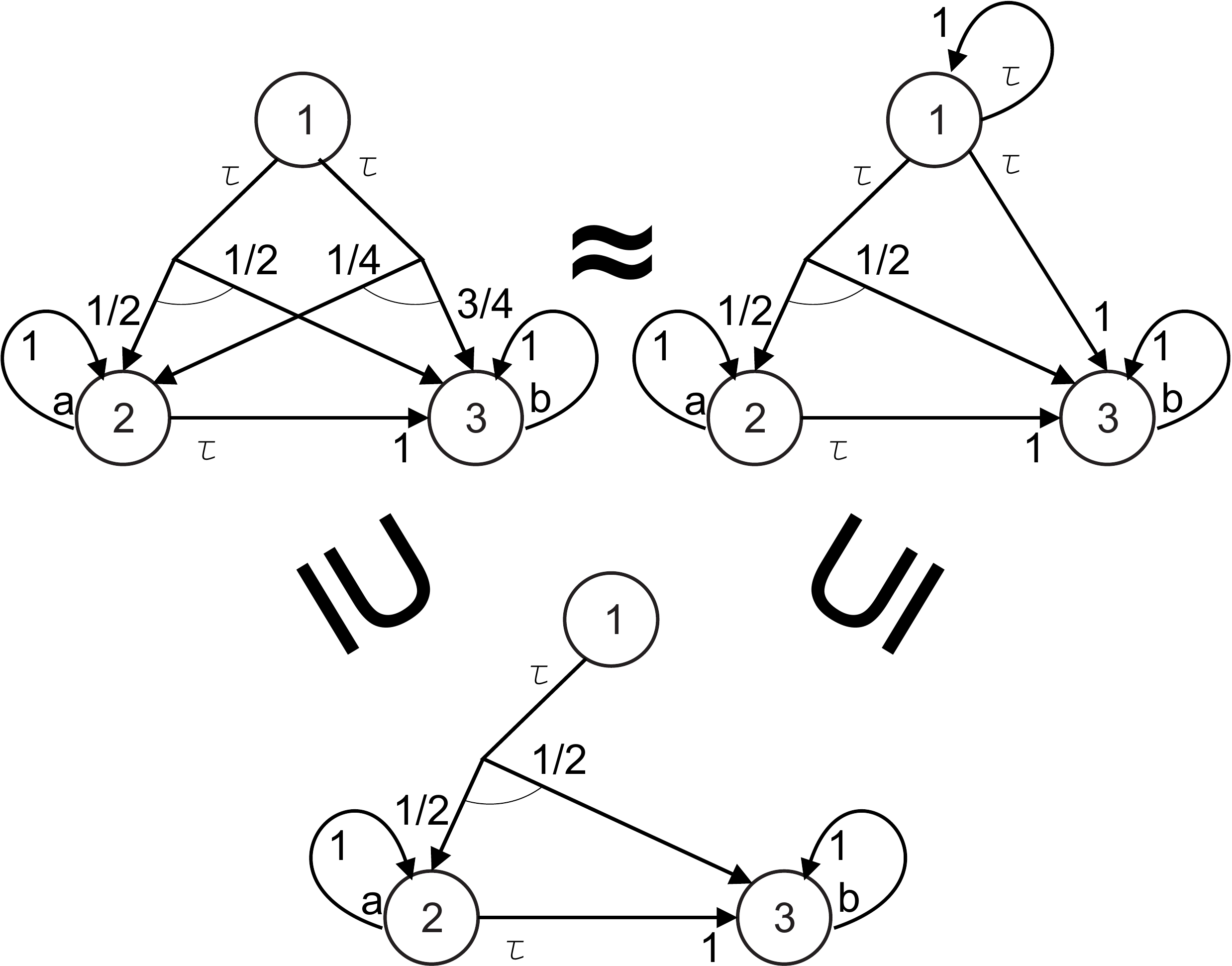}}
\caption{Intersection of bisimilar automata}
\label{fig:intersection_collection}
\end{center}
\end{figure}

The next example is given in Fig.~\ref{fig:weak_bisim}.
We start with the two weakly bisimilar automata at the top of the figure
Bisimilarity essentially stems from the following facts:
\ms{Hier den JS-Kommentar eingebaut:}
Firstly, the transition $(1,\tau, \Delta_3)$ of the PA on the right can be mimicked by the left automaton by a Dirac determinate scheduler that 
in state 1 always chooses $(1,\tau,\frac{1}{2}\Delta_2 \oplus \frac{1}{2}\Delta_3)$ and in state 2 always $(2,\tau,\Delta_3)$ and stops in state 3.
\ms{``stops in state 3'' besser weglassen?!}
Secondly, the transition
$(1,\tau,\frac{1}{4}\Delta_2 \oplus \frac{3}{4}\Delta_3)$
of the left automaton can be mimicked by the right automaton
by choosing each of the $\tau$-transitions emanating from state 1
with probability $\frac{1}{2}$. 
The intersection of both -- which is the canonical form -- is the automaton at the bottom of the figure.
\\

The next example shows that it is essential to have elements of $\nicefrac{\mathfrak{Q}_\sim (P)}{\someequiv}$, otherwise the intersection 
will not make sense. This example is given in Fig.~\ref{fig:wrong_weak_example}.
It is clear that e.g.~the distribution $\frac{3}{4}\Delta_2 \oplus \frac{1}{4} \Delta_3$ cannot be realised starting from state $1$
by the right automaton, but it can be realised by the left automaton, so both cannot be bisimilar.
Therefore, the intersection of both automata is not bisimilar to the left automaton.


\begin{figure}
\begin{center}
  \subfloat[Non-bisimilar automata]{\label{fig:wrong_weak_example}\includegraphics[width=4.7cm]{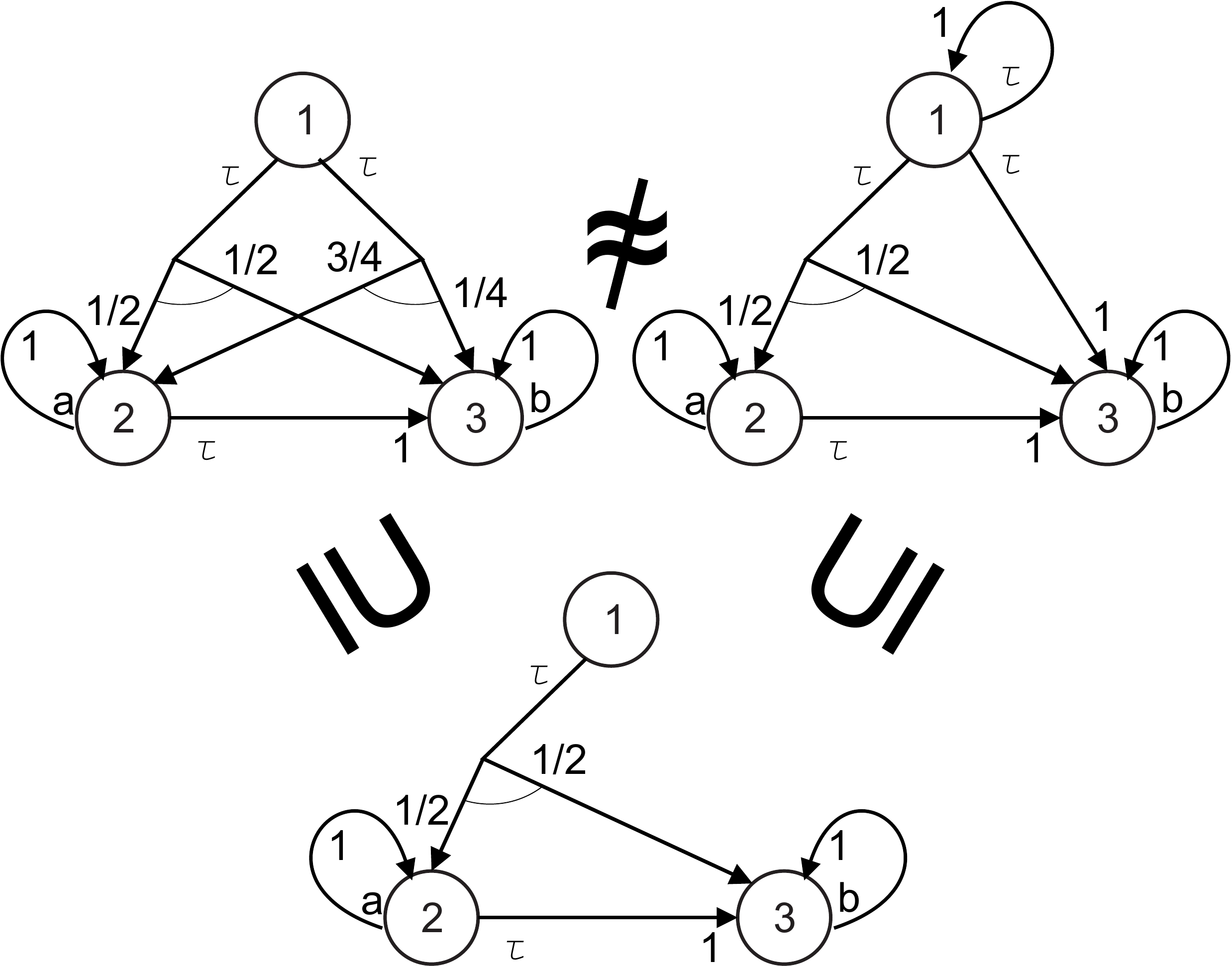}} \qquad
  \subfloat[The bounded case]{\label{fig:boundedcase}\includegraphics[width=8.5cm]{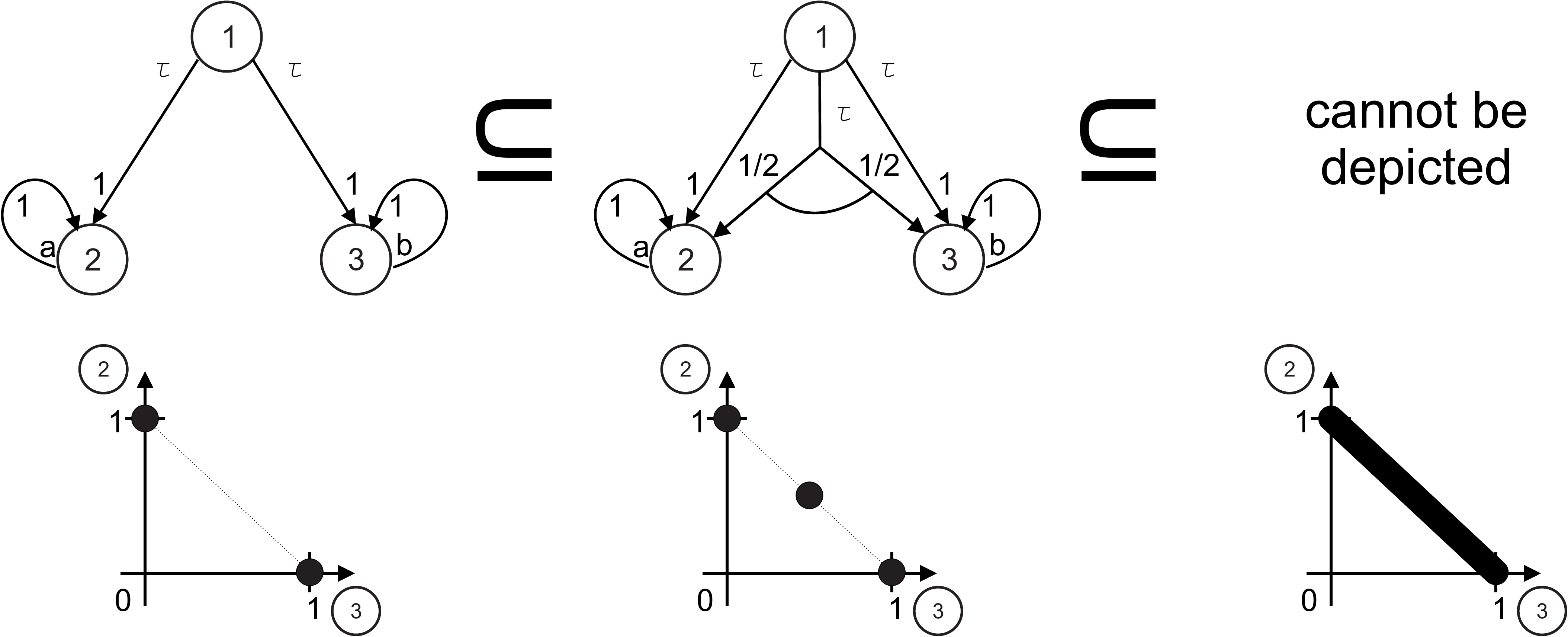}}
\caption{Examples continued}
\label{fig:some_collection}
\end{center}
\end{figure}

\subsection{A bounded example}
Assume that we consider strong bisimulation and want to calculate $\nicefrac{\mathfrak{Q}^*_\sim (P)}{\someequiv}$
for the automaton $P$ given in Fig.~\ref{fig:boundedcase} (middle). The least element (let's call it $\bot$) is shown
on the left, the greatest element cannot be adequately depicted, as it has 
(uncountably) infinitely many
transitions. The situation becomes clearer when (as suggested in \cite{segala:02}) considering
distributions as points in $\mathbb{R}^n$.
This is done in the lower line of the figure.
We only show the distributions that are possible by Dirac determinate schedulers starting from
state $1$. As there are only two successor states, $\mathbb{R}^2$ suffices for our purpose.
With this picture in mind it is clear that the greatest element is the PA given by
$\top=(\{1,2,3\}, \{\tau, a, b\}, \{(2,a,\Delta_2), (3,b,\Delta_3)\} \cup \{(1,\tau,c\Delta_2\oplus (1-c)\Delta_3) | c\in [0,1]\},1)$.
Clearly this is not a finite PA.
Generalising from this automaton, for $\mathfrak{A}\subseteq [0,1]$ we may define
$P_{\mathfrak{A}}:=(\{1,2,3\}, \{\tau, a, b\}, \{(2,a,\Delta_2), (3,b,\Delta_3)\} \cup \{(1,\tau,c\Delta_2\oplus (1-c)\Delta_3) | c\in \mathfrak{A}\},1)$.
Note that $\bot=P_{\{0,1\}}$, $\top = P_{[0,1]}$.
But now it is clear how the set $\nicefrac{\mathfrak{Q}^*_\sim (P)}{\someequiv}$ must look like:
$\nicefrac{\mathfrak{Q}^*_\sim (P)}{\someequiv}=\{P_{\mathfrak{A}} | \mathfrak{A}\subseteq [0,1], 0\in \mathfrak{A}, 1 \in \mathfrak{A} \}$
Leaving out $0$ or $1$ from $\mathfrak{A}$ will break bisimilarity.


\section{Conclusion}
\label{sec:conclusion}
This paper extends the notion of normal forms introduced in \cite{tacas:13} to the case of compact automata with countably infinite state space, countably
infinite set of actions and possibly uncountably many transitions.
We have justified the canonicity of normal forms by introducing them as
the
intersection of all bisimilar automata, not from
an abstract point of view as in \cite{tacas:13}.
The structure presented is nice as a theoretical result,
but (at least for the moment) there is no immediate practical applicability,
as the hard part is to construct the sets 
$\mathfrak{Q}_\sim (P)$ and $\mathfrak{Q}_\approx (P)$ where an uncountably infinite number of automata would have to be constructed. 
The structure itself is particularly nice for teaching purposes and for better under\-stand\-ing the possible 'shapes' of infinite automata.
\js{don't understand the meaning of this sentence.}
\ms{Wenn Du den Satz in Klammer meinst, der ist von Dir selbst ... Ich denke, gemeint ist, dass mann nur die Extrempunkte der konvexen Mengen benoetigt, oder?}
\js{Hab den Satz rausgenommen}
\ms{Auch geaendert, bitte lesen!}
\js{maybe we could add here (again), that constructing finite quotients has polynomial complexity according to the TACAS paper. So this part would not be too hard.}

\bibliographystyle{eptcs}
\bibliography{./local}

\end{document}